\documentclass[a4paper,11pt]{article}
\pdfoutput=1
\usepackage[T1]{fontenc}
\usepackage{lmodern}
\usepackage[margin=1.05in,includefoot]{geometry}
\usepackage{amsmath,amssymb,amsthm,mathtools}
\usepackage[colorlinks=true,linkcolor=blue,citecolor=blue,urlcolor=blue]{hyperref}
\usepackage[protrusion=false]{microtype}
\usepackage{pgfplots}
\usepackage{accsupp}
\pgfplotsset{compat=1.18}
\hypersetup{pdftitle={Exact value and rigidity of the 2x3 magic rectangle},
pdfauthor={Isaac Barouch Essayag; Aryeh Lev Zabokritskiy (Yohananov)},pdflang={en-US}}
\newtheorem{theorem}{Theorem}[section]
\newtheorem{lemma}[theorem]{Lemma}
\newtheorem{proposition}[theorem]{Proposition}
\newtheorem{corollary}[theorem]{Corollary}
\theoremstyle{remark}
\newcommand{\Id}{I}
\newcommand{\Tr}{\operatorname{Tr}}
\newcommand{\supp}{\operatorname{supp}}
\newcommand{\pg}{p_{\mathrm{guess}}}
\newcommand{\ket}[1]{\lvert #1\rangle}
\newcommand{\bra}[1]{\langle #1\rvert}
\newcommand{\ac}[2]{\{#1,#2\}}
\title{Exact value and rigidity of the $2\times3$ magic rectangle}
\newcommand{\XXThreeAuthors}{\begin{minipage}{0.88\textwidth}
\centering
\normalsize Isaac Barouch Essayag\par
\vspace{6pt}
\small MIGAL--Galilee Research Institute, Kiryat Shmona, Israel;\par
and Tel-Hai University of Kiryat Shmona in the Galilee,\par
Kiryat Shmona, Israel\par
\vspace{4pt}
{\hypersetup{hidelinks}\href{mailto:isaac.es@migal.org.il}{\texttt{isaac.es@migal.org.il}}}\par
\vspace{18pt}
\normalsize Aryeh Lev Zabokritskiy (Yohananov)\par
\vspace{6pt}
\small Department of Computer Science, MIGAL--Galilee Research Institute,\par
Kiryat Shmona, Israel; and Department of Computer Science,\par
Tel-Hai University of Kiryat Shmona in the Galilee, Kiryat Shmona, Israel\par
\vspace{4pt}
{\hypersetup{hidelinks}\href{mailto:yuhanalev@telhai.ac.il}{\texttt{yuhanalev@telhai.ac.il}}}\par
\vspace{4pt}
ORCID: {\hypersetup{hidelinks}\href{https://orcid.org/0000-0003-3151-6192}{0000-0003-3151-6192}}
\end{minipage}}
\author{\XXThreeAuthors}
\date{}
\begin{document}
\maketitle
\begin{abstract}
\leavevmode\pdfbookmark[1]{Abstract}{abstract}%
The $2\times3$ magic rectangle is a nonlocal game in which two players
match entries subject to incompatible parity constraints. Its quantum value
was suggested by numerical bounds and known quantum random access codes.
We prove that this value is $(1+\sqrt{2/3})/2$ and classify all optimal
strategies in arbitrary finite local dimensions, allowing general measurements:
up to local isometries and ancillary systems, each contains a maximally
entangled pair of four-dimensional systems with fixed measurements on the
local state supports. The same value holds for commuting-operator strategies
and for the almost-quantum relaxation, confirming the Adamson--Wallden
conjecture for this rectangle. For finite-dimensional magic-square strategies,
the largest score compatible with perfect prediction of Alice's complete
first row by an adversary with quantum side information is one third plus
two thirds of the rectangle value. We classify the strategies attaining
this threshold with perfect prediction and give linear-order upper and lower
bounds on one minus the best compatible guessing probability above it.
The rectangle value is also the optimal average success probability of
finite-dimensional three-bit random access codes that hide the full input
parity, and we classify all optimal codes.
\end{abstract}

\section{Introduction}
Nonlocal games test quantum correlations and can identify the state and
measurements responsible for an optimal score. Magic rectangles form a
concrete family extending the magic square of Mermin and Peres
\cite{Mermin,Peres}: one player supplies a row and the other a column,
subject to incompatible parity constraints. Adamson and Wallden
characterized many rectangle values and connected them to randomness
expansion \cite{AW}. Three-row rectangles supplemented with check rounds
also self-test Bell states in parallel \cite{AWselftest}.

For the $2\times3$ rectangle, the quantum value remained undetermined.
Its classical value is $5/6$. Adamson and Wallden obtained the lower bound
$(4+\sqrt2)/6\approx0.9024$ and a numerical upper bound with conjectured
closed form $(1+\sqrt{2/3})/2\approx0.9082$. The upper bound comes from the
almost-quantum semidefinite relaxation \cite{almostquantum}; their conjecture
covers all two-row rectangles \cite[Theorem~13, Table~II, Conjecture~17]{AW}.

Measurements attaining the proposed value are available in the quantum
random access codes of Imamichi and Raymond \cite{IR} and Suzuki
\cite{Suzuki}. Through remote preparation, these constructions supply the
attaining strategy used here (Section~\ref{sec:rectangle}). Man\v{c}inska and
Storgaard proved the same optimal score for
three-bit codes on two qubits \cite{MS}; general bounds for prescribed
encoding dimensions are developed in \cite{FMT}. These coding results fix
the encoding dimension. The rectangle problem allows arbitrary local
dimensions and asks both for the optimal score and for the structure of
all strategies attaining it.

We prove that the quantum value is exactly $(1+\sqrt{2/3})/2$ and classify
all optimal strategies in arbitrary finite local dimensions, allowing
general local measurements. Every optimal strategy contains a maximally
entangled pair of four-dimensional systems, decoupled from any system held
by an adversary. The classification also identifies the measurements on the
local state supports, up to local isometries and ancillary systems. Thus, at
the optimum, the rectangle self-tests both the entangled state and the
measurements. The same value holds for commuting-operator strategies and
for the almost-quantum relaxation \cite{almostquantum}. This confirms
the conjectured expression for the $2\times3$ case.

The rectangle value also answers a randomness question for the magic
square: what is the largest score compatible with perfect prediction of
Alice's complete first row by an adversary with quantum side information?
This question was left open in \cite[Appendix C]{CRS}.
For finite-dimensional strategies in a single round, we obtain the exact
threshold $(2+\sqrt{2/3})/3\approx0.939$: a perfectly predicted row can be
won with certainty, and the other two rows then form the rectangle.
Adamson and Wallden's score $(7+\sqrt2)/9\approx0.935$ concerns a different
condition: both devices are deterministic at one distinguished input pair
\cite[Theorem~20]{AW}. Our condition requires perfect prediction of Alice's
row. We also classify all finite-dimensional strategies attaining our
threshold with perfect prediction. Every such strategy still contains a
maximally entangled pair of four-dimensional systems, even when the
adversary knows the entire row.

Above the threshold, we bound one minus the best compatible guessing
probability, and the smallest compatible conditional min-entropy and
conditional von Neumann entropy, above and below by constant multiples of
the score excess. The bounds cover scores up to one. The von Neumann entropy
lower bound also applies to the parity-extended magic-square games of
\cite{CRS}, for any number of parties and finite-dimensional local systems,
by the argument of \cite[Lemma~6.6]{CRS}.
It improves the numerical bound of \cite[Lemma~6.7]{CRS} between our
threshold and a score of about $0.943$; the earlier bound is stronger
above that value.

A second application concerns three-bit random access codes that hide
the full input parity, while two-bit parities may be revealed
\cite{CKKS,RP}. Their optimal success probability equals the rectangle
value, resolving the three-bit odd-parity-oblivious case left open in
\cite{CKKS}. This score exceeds the value stated to be optimal for
four-dimensional encodings in \cite[arXiv v1]{RP}. We classify all optimal
finite-dimensional encodings and their decoders on the occupied support.
For approximate parity hiding, we give an upper bound on success with a
correction linear in the trace distance between the two parity averages
and independent of the encoding dimension.

Our contributions are:
\begin{itemize}
\item We determine the rectangle value and classify all optimal strategies
in arbitrary finite local dimensions.
\item We determine the magic-square threshold for perfect prediction of
Alice's first row, classify the finite-dimensional strategies attaining it
with perfect prediction, and establish the linear order of one minus the
best compatible guessing probability above the threshold.
\item We determine the optimal success probability and classify the optimal
finite-dimensional codes hiding the full three-bit parity, and bound
success under approximate parity hiding.
\end{itemize}
The paper is organized as follows. Section~\ref{sec:rectangle} defines the
rectangle and the operator model, presents the four-dimensional attaining
strategy and states the main theorem. Section~\ref{sec:sos} proves the value
bound, including its commuting-operator and almost-quantum versions, and
Section~\ref{sec:rigidity} classifies the optimal strategies.
Section~\ref{sec:threshold} determines the magic-square threshold for perfect
prediction and classifies the strategies attaining it.
Section~\ref{sec:guessing} bounds the guessing probability and conditional
entropies above the threshold. Section~\ref{sec:rac} treats random access
codes that hide the full input parity, and Section~\ref{sec:conclusion} concludes.

\section{The rectangle and its optimal strategies}\label{sec:rectangle}
The rectangle retains the parity conflict of the magic-square game while
allowing only two row inputs. Its score will be expressed as the expectation
of one Bell operator. This form makes the upper bound and the equality
conditions accessible to the same algebraic argument.

Alice receives a row $x\in\{0,1\}$, the first or the second row, and outputs an even sign string of length
three. Bob receives a column $j\in\{1,2,3\}$ and outputs two opposite signs, one for each row.
The inputs are independent and uniform, and the players win when their
intersection signs agree. We write $\omega_Q(2,3)$ for the supremum of the
winning probability over tensor-product strategies of arbitrary finite
local dimensions, with arbitrary local positive-operator-valued measures
(POVMs). Purifications may include an arbitrary system $E$ held by Eve.
The commuting-operator model instead uses one Hilbert space with commuting
Alice and Bob measurement algebras; there, too, general measurements may be
taken projective \cite[Sec.~2.2]{NPA}.

Negate Alice's second output row, so that its sign product becomes $-1$,
and denote its coordinate observables by $C_1,C_2,C_3$. Denote the first-row
observables by $D_1,D_2,D_3$, and Bob's first sign in column $j$ by $B_j$.
For a finite family of POVMs, one fixed local isometric embedding
and input-dependent projective measurements realize all the original
measurements. Their outcome labels enforce the required parities.
This reduction preserves all correlations and uses only local ancillas.
After this local projective dilation, we consider arbitrary Hermitian
reflections $D_j$, $C_j$ and $B_j$ satisfying
\begin{equation}\label{eq:relations}
 D_j^2=C_j^2=B_j^2=\Id,\qquad
 D_1D_2D_3=\Id,\qquad C_1C_2C_3=-\Id.
\end{equation}
The $D$ operators commute with one another, as do the $C$ operators.
Every Alice operator commutes with every Bob operator.

With tensor factors suppressed, put
\begin{equation}\label{eq:notation}
r\coloneqq\sqrt6,\qquad q\coloneqq\frac{2r}{3}=\sqrt{\frac83},\qquad
X_j\coloneqq D_j+C_j,\qquad S\coloneqq\sum_{j=1}^3X_jB_j.
\end{equation}
The rectangle score is $1/2+\langle S\rangle/12$. The sign change in the
second row makes both rows contribute positively to $X_jB_j$: Bob's actual
second sign is $-B_j$. Thus an upper bound on $S$ bounds the original game.

We next define the reference strategy. Fix the even strings in the order
\begin{equation}\label{eq:alpha}
\alpha_1\coloneqq(1,1,1),\quad \alpha_2\coloneqq(1,-1,-1),\quad
\alpha_3\coloneqq(-1,1,-1),\quad \alpha_4\coloneqq(-1,-1,1),
\end{equation}
and set
\begin{equation}\label{eq:matrix}
M\coloneqq\begin{pmatrix}0&1&1&1\\-1&0&1&-1\\-1&-1&0&1\\-1&1&-1&0\end{pmatrix},
\qquad U_0\coloneqq M/\sqrt3.
\end{equation}
Here $MM^T=3\Id$; such a matrix with zero diagonal and signs elsewhere is
a conference matrix. With the zero diagonal, and since each coordinate takes each sign twice among the $\alpha_a$, $M$ gives the identities \eqref{eq:construction} below. A superscript $0$ marks reference objects. On $\mathbb C^4$ let
\begin{align}\label{eq:ideal}
D_1^0&\coloneqq\operatorname{diag}(1,1,-1,-1),&
D_2^0&\coloneqq\operatorname{diag}(1,-1,1,-1),&
D_3^0&\coloneqq\operatorname{diag}(1,-1,-1,1),\notag\\
C_j^0&\coloneqq-U_0D_j^0U_0^*,&
X_j^0&\coloneqq D_j^0+C_j^0,& B_j^0&\coloneqq(X_j^0)^T/q.
\end{align}
Alice's row measurements have rank-one projections
$P_a^0\coloneqq\ket a\bra a$ and
$Q_a^0\coloneqq U_0\ket a\bra aU_0^*$, respectively, both labeled by $\alpha_a$.
Her physical second-row observables are $-C_j^0$. The identities
\begin{equation}\label{eq:construction}
\ac{D_j^0}{C_j^0}=\frac23\Id,\qquad
(X_j^0)^2=\frac83\Id
\end{equation}
make Bob's observables reflections. On
$\ket{\Phi_4}\coloneqq\frac12\sum_{a=1}^4\ket{a,a}$,
the relation
$\langle A\otimes B\rangle_{\Phi_4}=\Tr(AB^T)/4$
shows that each intersection correlation is $\sqrt{2/3}$, so
$\langle S\rangle_{\Phi_4}=6\sqrt{2/3}=2r$.
By remote preparation, the reference strategy realizes, up to a unitary change of basis, the known random access code
construction using two orthonormal bases, one for each input parity
\cite{IR,Suzuki}; Section~\ref{sec:rac} makes this correspondence explicit, and Theorem~\ref{thm:RACclass} implies the unitary equivalence.
The following theorem proves that this strategy is optimal and describes
all finite-dimensional strategies attaining its score.

\begin{theorem}[Exact value and rigidity]\label{thm:rigidity}
The quantum and commuting-operator values of the uniform $2\times3$ magic
rectangle are
\begin{equation}\label{eq:value}
\omega_Q(2,3)=\frac{1+\sqrt{2/3}}2.
\end{equation}
Every optimal finite-dimensional tensor-product strategy, with arbitrary
POVMs and a purification $\psi_{ABE}$, admits isometries $V_A$ and $V_B$ of Alice's and
Bob's local spaces into $\mathbb C^4_{A_0}\otimes\mathcal H_{A'}$ and $\mathbb C^4_{B_0}\otimes\mathcal H_{B'}$ under which
\begin{equation}\label{eq:extract}
\psi_{ABE}\longmapsto
\ket{\Phi_4}_{A_0B_0}\otimes\ket\eta_{A'B'E}.
\end{equation}
The original measurement effects leave the local reduced-state supports
invariant, and $V_A$ and $V_B$ map these supports onto $\mathbb C^4\otimes\mathcal K_A$ and $\mathbb C^4\otimes\mathcal K_B$ for subspaces $\mathcal K_A\subseteq\mathcal H_{A'}$ and $\mathcal K_B\subseteq\mathcal H_{B'}$. There, the conjugated effects $V_A(\cdot)V_A^*$ of Alice's rows $0$ and $1$ for the output $\alpha_a$ are
$P_a^0\otimes\Id$ and $Q_a^0\otimes\Id$, and Bob's conjugated observable for his first sign in column $j$ is
$B_j^0\otimes\Id$. Conversely, every strategy that takes this form under some local
isometries is optimal. Each local reduced-state
support has dimension divisible by four; in particular, the smallest local
dimension of an optimal strategy is four.
\end{theorem}

The upper bound is proved in Section~\ref{sec:sos}. Equality makes the
positive terms in that proof vanish. As in the general approach developed in,
for example, \cite{Cui}, Section~\ref{sec:rigidity} translates
these vanishing relations into a classification of Alice's measurement
algebra and then constructs the required isometry on Bob's system. This
last step keeps Eve in the ancillary factor and completes the return from
operator relations to the original local strategy.
\section{An exact sum-of-squares bound}\label{sec:sos}
To bound the rectangle score in arbitrary dimension, we express the gap
$2r\Id-S$ as a positive sum of operator squares. The terms are chosen so
that equality also relates each Alice observable to a Bob-local operator;
these relations will determine the optimal state and measurements. The
first two families of squares below carry these relations, the second after substituting those of the first; the third
completes the identity.
Put $Y_j\coloneqq D_j-C_j$ and $Z_j\coloneqq X_jB_j$.
For each $j$, let $i,k$ be the other two indices, in either order, and put
\begin{equation}\label{eq:VW}
V_j\coloneqq X_iB_k+X_kB_i,\qquad W_j\coloneqq Y_iB_k+Y_kB_i.
\end{equation}
\begin{theorem}[Exact certificate]\label{thm:SOS}
For Hermitian operators satisfying \eqref{eq:relations} and the commutation relations stated after it,
\begin{align}\label{eq:SOS}
2r\Id-S={}&\frac{5r}{64}\sum_j(Z_j-q\Id)^2
+\frac r{12}\sum_j(Y_j-V_j/r)^2\notag\\
&+\frac r{192}\sum_j(X_j+qB_j-qW_j)^2.
\end{align}
Consequently $S\leq2\sqrt6\Id$ in every commuting-operator representation.
\end{theorem}
\begin{proof}
Every displayed residual is Hermitian, and all coefficients are positive.
We verify the identity directly. For distinct $i,k,j$, the row relations give
\begin{equation}\label{eq:rowcancel}
X_j^2+Y_j^2=4\Id,\quad
X_iX_k+Y_iY_k=2Y_j,\quad
X_iY_k+Y_iX_k=2X_j.
\end{equation}
Define $J\coloneqq\sum_jY_j\ac{B_i}{B_k}$. Expanding ordered products and using
\eqref{eq:rowcancel} gives
\begin{align}
\sum_j(V_j^2+W_j^2)&=24\Id+2J,\label{eq:cancela}\\
\sum_j\ac{W_j}{B_j}&=2J,\label{eq:cancelb}\\
\sum_j\ac{Y_j}{V_j}&=4S,
\qquad \sum_j\ac{X_j}{W_j}=4S.\label{eq:cancelc}
\end{align}
For example, in the first sum each $X_j^2+Y_j^2$ appears twice;
the remaining term for the complementary pair is $2Y_j\ac{B_i}{B_k}$.
In the last two identities, collect the coefficient of each $B_j$ and use
the third identity of \eqref{eq:rowcancel} in both orders.

Multiply the right side of \eqref{eq:SOS} by $192/r$.
Its coefficients become $15,16,1$. The terms containing
$\sum X_j^2$ and $\sum Y_j^2$ sum to $192\Id$;
the scalar terms from $q\Id$ and $qB_j$ sum to $48q^2\Id=128\Id$.
Since $16/r^2=q^2=8/3$, the $V,W$ squares add $(8/3)(24\Id+2J)$, and the
$B,W$ cross terms subtract $(8/3)2J$. The Bell terms have coefficients
$-30q$ from the first family, $-64/r$ from the second and $2q-4q$ from the
third; since $q=2r/3$ and $r^2=6$, they total $-32r$. The result is $384\Id-32rS$,
which equals $(192/r)(2r\Id-S)$.
\end{proof}

Together with the reference strategy, the certificate proves the value
assertion of Theorem~\ref{thm:rigidity}. For an approximately optimal strategy
it also gives the following robustness estimates for the relations that
drive the exact equality analysis.

\begin{corollary}[Residual estimate]\label{cor:residual}
For operators as in Theorem~\ref{thm:SOS} and a unit vector $\psi$, let $\gamma\coloneqq2r-\langle S\rangle_\psi$, which is
twelve times the score deficit from the optimum.
Then
\begin{align*}
\sum_j\|(X_j-qB_j)\psi\|^2&\leq\frac{64}{5r}\gamma,\\
\sum_j\|(Y_j-V_j/r)\psi\|^2&\leq\frac{12}{r}\gamma.
\end{align*}
These estimates control, in the norm $\|\cdot\,\psi\|$, the relations $X_j\psi=qB_j\psi$ and $Y_j\psi=V_j\psi/r$; at the optimum both hold exactly, and substituting the first into the second gives the Bob-local form $Y_j\psi=\tfrac23\ac{B_i}{B_k}\psi$.
\end{corollary}
\begin{proof}
Drop the other nonnegative square families in \eqref{eq:SOS}. Multiplication
by the reflection $B_j$ preserves the first residual norm.
\end{proof}

The same identity also determines the value of a standard semidefinite
relaxation. In the Navascu\'es--Pironio--Ac\'in (NPA) hierarchy
\cite{NPA}, level $1+AB$ uses words containing at most one outcome
projector from each party as Gram-matrix indices. For
two parties this is the almost-quantum set \cite{almostquantum}. Each row is
one four-outcome measurement whose outcomes are the four legal, that is
parity-respecting, sign strings; their orthogonality supplies the parity
reductions at degree two (Appendix~\ref{app:npa}).

\begin{corollary}[Exact almost-quantum value]\label{cor:almost}
For the rectangle's legal outcome measurements, NPA level $1+AB$ has value
\[
\omega_{1+AB}(2,3)=\omega_Q(2,3)=\frac{1+\sqrt{2/3}}2.
\]
\end{corollary}
\begin{proof}
Expand the coordinate observables in their outcome projectors. Each
residual in \eqref{eq:SOS} belongs to the $1+AB$ word span, and every
identity used in its expansion follows from outcome orthogonality and
completeness within words containing at most two letters from each party.
The detailed moment constraints are given in Appendix~\ref{app:npa}.
Consequently the moment functional sends the right side of
\eqref{eq:SOS} to a sum of nonnegative Gram quadratic forms and gives
the same upper bound. The reference quantum strategy supplies an attaining
feasible Gram matrix.
\end{proof}

The natural alphabets of \cite[Sec.~3]{AW} consist of the legal strings, so their level-$1+AB$ relaxation is the one above, and this proves the case $n=3$ of the almost-quantum expression conjectured for $2\times n$ rectangles in
\cite[Table II, Eq.~(34), and Conjecture 17]{AW}. For comparison,
Appendix~\ref{app:trace} gives a shorter proof for maximally entangled
states and an example explaining the limitation of a direct operator
extension of that argument. Appendix~\ref{app:second} records a second
sum-of-squares certificate for the same Bell operator, whose residuals vanish on the same optimal states.
\section{Rigidity of optimal strategies}\label{sec:rigidity}
We prove the classification in Theorem~\ref{thm:rigidity}. The argument
answers two questions: which row measurements can saturate the certificate,
and how their relations determine the state up to separate local isometries.
First, vanishing squares make the reduced-state supports invariant. The row
measurements on these supports then generate one copy of the full
$4\times4$ matrix algebra, with an arbitrary multiplicity.
\begin{lemma}[Support relations]\label{lem:support}
Let a finite-dimensional strategy with projective measurements satisfying \eqref{eq:relations},
with a purification $\psi_{ABE}$, attain
$\langle S\rangle=2r$. Then $[\rho_A,D_j]=[\rho_A,C_j]=0$, and on Alice's reduced-state support the two row triples
obey
\begin{equation}\label{eq:support}
\ac{D_j}{C_j}=\tfrac23\Id.
\end{equation}
In particular, the restrictions to that support are again internally commuting reflection triples with products $+\Id$ and $-\Id$.
\end{lemma}
\begin{proof}
The first and second square families in \eqref{eq:SOS} imply the transfer relations
\begin{equation}\label{eq:transfer}
X_j\psi=qB_j\psi,\qquad
Y_j\psi=\frac1r(X_iB_k+X_kB_i)\psi
=\frac23\ac{B_i}{B_k}\psi.
\end{equation}
The right sides are Hermitian Bob-local operators. Whenever Hermitian
operators satisfy $L_A\psi=R_B\psi$, taking a partial trace of this relation
and its adjoint gives $L_A\rho_A=\rho_A L_A$.
Thus $X_j,Y_j$, and hence $D_j,C_j$, preserve $\supp\rho_A$.
Squaring the first transfer, using cross-party commutation and $B_j^2=\Id$,
gives $X_j^2\psi=q^2\psi$. An Alice operator $L$ with $L\psi=0$ vanishes
on $\supp\rho_A$, because $\|L\psi\|^2=\Tr(L^*L\rho_A)$; this proves the anticommutator identity there.
\end{proof}

The remaining classification is entirely on Alice's support. The
anticommutator relations fix the overlaps between the two row measurements.
We use those overlaps to identify their joint algebra before returning to
the shared state.
For scalar blocks ($m=1$ below), the conference-matrix normal form obtained
in the proof is the order-four
case of the unit weighing matrix classification in
\cite[Theorem 12]{BKR}. The argument below treats unitary blocks and
identifies the generated algebra.

\begin{lemma}[Classification of the row algebra]\label{lem:conference}
Every finite-dimensional pair of internally commuting reflection triples with products
$+\Id,-\Id$ and $\ac{D_j}{C_j}=2\Id/3$ is unitarily equivalent to
the reference triples $D_j^0,C_j^0$ tensored with $\Id_m$.
The underlying space has dimension $4m$, and the algebra generated by the two triples is
$M_4(\mathbb C)\otimes\Id_m$.
\end{lemma}
\begin{proof}
The joint spectral projections are
\[
P_a=\tfrac14(\Id+\sum_j\alpha_{aj}D_j),\qquad
Q_b=\tfrac14(\Id-\sum_j\alpha_{bj}C_j),
\]
where the even sign strings satisfy
$\alpha_a\cdot\alpha_b=4\delta_{ab}-1$.
Compressing the anticommutator identities to the two spectral subspaces gives
\[
P_aC_jP_a=\frac{\alpha_{aj}}3P_a,\qquad
Q_bD_jQ_b=-\frac{\alpha_{bj}}3Q_b.
\]
Consequently,
\begin{equation}\label{eq:overlap}
P_aQ_bP_a=\begin{cases}0,&a=b,\\P_a/3,&a\ne b,\end{cases}
\qquad
Q_bP_aQ_b=\begin{cases}0,&a=b,\\Q_b/3,&a\ne b.\end{cases}
\end{equation}
The diagonal identities imply $P_aQ_a=0$.
For $a\ne b$, the map $\sqrt3\,P_a|_{Q_b\mathcal H}$ is a unitary
from $Q_b\mathcal H$ onto $P_a\mathcal H$: its adjoint products are
$3Q_bP_aQ_b=Q_b$ and $3P_aQ_bP_a=P_a$.
For any two $P$ ranges, choose $b$ different from their indices; both
are unitarily equivalent to $Q_b\mathcal H$. Thus all eight spectral
subspaces have a common dimension $m$, and the full space has dimension $4m$.

Choose orthonormal bases for the four $P$ ranges and four $Q$ ranges. Their
change-of-basis unitary $U$, with rows indexed by the $P$ bases and columns by
the $Q$ bases, has $m\times m$ blocks: diagonal blocks are zero and every
off-diagonal block is a unitary divided by $\sqrt3$.
Block-diagonal changes of bases make the first row and first column of
$\sqrt3 U$ equal to identity blocks off the diagonal. Write the result as
\[
\begin{pmatrix}
0&\Id&\Id&\Id\\\Id&0&A&B\\\Id&C&0&D\\\Id&E&F&0
\end{pmatrix}.
\]
Orthogonality with the first row gives $B=-A,D=-C,F=-E$;
orthogonality with the first column gives $E=-C,F=-A$.
Thus $C=-A,D=A,E=A,F=-A$. Multiplying block row 1 on the left by $A$
and block columns 2,3,4 on the right by $A^*$ removes the remaining unitary.
Negating the first column of the resulting scalar conference matrix gives
\eqref{eq:matrix}.

Each of these gauges only changes the chosen bases inside the ranges of the
$P_a$ or of the $Q_b$, so all eight projections are unchanged.
In the resulting $P$-adapted basis, $P_a=P_a^0\otimes\Id_m$ and
$Q_b=(U_0\otimes\Id_m)(P_b^0\otimes\Id_m)(U_0\otimes\Id_m)^*=Q_b^0\otimes\Id_m$;
hence $D_j=D_j^0\otimes\Id_m$ and $C_j=C_j^0\otimes\Id_m$.

Finally, for $a\ne c$, choose $b$ distinct from both.
In the scalar form, $P_aQ_bP_c$ is a nonzero scalar multiple of
$\ket a\bra c\otimes\Id_m$. Together with the $P_a$, these supply all
matrix units. This proves the algebra assertion.
\end{proof}

The row algebra now gives an Alice-local four-dimensional factor. To locate the
matching factor on Bob, transfer words in Alice's generators to Bob in
reverse order. The resulting matrix units identify a Bob-local four-dimensional factor
and force the state to be maximally entangled on the two extracted factors.

\begin{proof}[Proof of Theorem~\ref{thm:rigidity}]
The value follows from Theorem~\ref{thm:SOS} and the construction in
Section~\ref{sec:rectangle}. For the equality assertion, first dilate the POVMs locally. Lemmas \ref{lem:support} and
\ref{lem:conference} identify Alice's support with
$\mathbb C^4\otimes\mathcal K_A$, with the reference operators $D_j^0,C_j^0$ on
the first factor. The first transfer in \eqref{eq:transfer} also implies
$[\rho_B,B_j]=0$, so Bob's support is invariant.

The $X_j$ alone generate the row algebra. Indeed the transfers imply,
on Alice's support,
\begin{equation}\label{eq:Yrecover}
Y_j=\tfrac14\ac{X_i}{X_k}.
\end{equation}
To verify this, note that the first transfer relation and cross-party
commutation move every word in the $X_j$ to Bob in reverse order:
\begin{equation}\label{eq:words}
X_{j_1}\cdots X_{j_n}\psi=q^nB_{j_n}\cdots B_{j_1}\psi.
\end{equation}
Apply this to $X_iX_k+X_kX_i$, use $q^2=8/3$, and compare with the second
relation of \eqref{eq:transfer}. Thus $D=(X+Y)/2$ and $C=(X-Y)/2$ lie in the
algebra generated by $X$.

For a polynomial in the $X_j$, transfer each word to Bob by \eqref{eq:words}.
If a polynomial vanishes on Alice's support, its reversed Bob polynomial
annihilates $\psi$ and hence vanishes on Bob's support.
This gives a well-defined unital $*$-antihomomorphism
\[
\Theta:M_4(\mathbb C)\otimes\Id_{\mathcal K_A}\longrightarrow\mathcal B(\supp\rho_B),
\qquad \Theta(X_j)\coloneqq qB_j.
\]
It is injective: a zero Bob image annihilates $\psi$, so its Alice
preimage vanishes on Alice's support. Put $F_{ab}\coloneqq\Theta(E_{ba}\otimes\Id_{\mathcal K_A})$ for the
usual matrix units $E_{ab}$. Then
$F_{ab}F_{cd}=\delta_{bc}F_{ad}$, $F_{ab}^*=F_{ba}$, and
$\sum_aF_{aa}=\Id$. Choosing a basis of $\operatorname{ran}F_{11}$ and
transporting it with the $F_{a1}$ gives a Bob-only identification
$\supp\rho_B\simeq\mathbb C^4\otimes\mathcal K_B$.

By \eqref{eq:words} and linearity, $x\psi=\Theta(x)\psi$ for every $x$ in
the row algebra. In particular, every matrix unit satisfies
$(E_{ab}\otimes\Id_{\mathcal K_A})\psi=(E_{ba}\otimes\Id_{\mathcal K_B})\psi$, where the left operator acts on Alice's system and the right on Bob's.
The diagonal relations restrict the $\mathbb C^4\otimes\mathbb C^4$ components of $\psi$ to
$\sum_a\ket{a,a}\otimes\eta_a$; the off-diagonal relations give
$\eta_a=\eta_b$ for all $a,b$. Normalization proves \eqref{eq:extract}.
Bob's isometry is constructed from operators on Bob's system alone.
Under these identifications $\Theta(E_{ab}\otimes\Id_{\mathcal K_A})=E_{ba}\otimes\Id_{\mathcal K_B}$ for all $a,b$. Since $\Theta(X_j)=qB_j$ and $X_j=X_j^0\otimes\Id_{\mathcal K_A}$, linearity gives $qB_j=(X_j^0)^T\otimes\Id_{\mathcal K_B}$, so Bob's observables have the claimed form.

The dilated reduced-state supports are the images of the original supports.
Their invariance under the dilated row and column projectors implies
that the original POVM effects leave the original supports invariant and form the
corresponding projective measurements on them: for a dilation embedding $J$ and a dilated projector
$\Pi$, the original effect is $J^*\Pi J$, and $\Pi J=J(J^*\Pi J)$ there. Composing the extraction isometries with the local
dilation embeddings proves the claim for the original strategy.
Mapping the orthogonal complements of the original supports isometrically
into $\mathbb C^4\otimes(\mathcal H_{A'}\ominus\mathcal K_A)$ and $\mathbb C^4\otimes(\mathcal H_{B'}\ominus\mathcal K_B)$, for large enough $\mathcal H_{A'}\supseteq\mathcal K_A$ and $\mathcal H_{B'}\supseteq\mathcal K_B$, extends these maps to isometries of the
original local spaces.
Converse attainability follows from Section~\ref{sec:rectangle}. The dimension assertions follow from $\supp\rho_A\simeq\mathbb C^4\otimes\mathcal K_A$, $\supp\rho_B\simeq\mathbb C^4\otimes\mathcal K_B$, the fact that the dilated supports are isometric images of the original ones, and the reference strategy.
\end{proof}

The extracted state also determines the observed probabilities and their
privacy against Eve. Here the row label $a$ denotes the physical output
$\alpha_a$ for either row, and $b$ is Bob's first sign.

\begin{corollary}[Optimal correlations and private randomness]
\label{cor:rectangle-randomness}
Every optimal finite-dimensional rectangle strategy has
\[
p(a,b\mid x,j)=\frac18
\left(1+(-1)^x\sqrt{\frac23}\,b\,\alpha_{aj}\right),
\qquad x\in\{0,1\},\quad j\in\{1,2,3\}.
\]
For either public row input $x$, Alice's row is uniform on its four legal
outputs and independent of Eve. In particular,
\[
\pg(A_{\mathrm{row}}\mid E,x)=\frac14,
\qquad H_{\min}(A_{\mathrm{row}}\mid E,x)=2.
\]
\end{corollary}
\begin{proof}
Apply Theorem~\ref{thm:rigidity}. Alice's row effects, restricted to her support, are the
rank-one reference projections tensored with the identity, and the state
is $\Phi_4\otimes\eta$. Hence each subnormalized conditional Eve state
equals $\rho_E/4$, proving uniformity, independence and the guessing claim.
For the joint distribution, put $R_{a|0}\coloneqq P_a^0$ and $R_{a|1}\coloneqq Q_a^0$.
Taking diagonal entries of the anticommutator in \eqref{eq:construction} in the two bases gives
\[
\Tr\bigl(R_{a|x}(B_j^0)^T\bigr)
=(-1)^x\sqrt{\frac23}\,\alpha_{aj}.
\]
Taking the maximally entangled expectation of
$R_{a|x}\otimes(\Id+bB_j^0)/2$ gives the displayed probability.
\end{proof}

This randomness concerns the measured rows of the rectangle. The extra
deterministic row used below to embed the rectangle in the magic square
is a different measurement.

\section{The magic-square threshold for perfect prediction}\label{sec:threshold}
The rectangle value determines when the full magic-square score first
certifies uncertainty about Alice's first-row output against quantum side information.
The relevant restriction is perfect prediction of Alice's key. Conditioning
on a successful prediction makes that row deterministic, and a local
postprocessing of Bob's columns then leaves the rectangle game on the other two rows.
In the full game Alice receives a row $x\in\{0,1,2\}$ and answers three
signs of product $+1$; Bob receives a column $j\in\{1,2,3\}$ and answers
three signs of product $-1$. The inputs are independent and uniform, and
the players win when Alice's $j$th sign equals Bob's sign in row $x$. Let the key $K$
be Alice's entire output row on input $x=0$. Over finite-dimensional
tensor-product strategies with arbitrary POVMs and quantum side
information, let $\beta_0$ be the supremum of the score subject to
$\pg(K\mid E,x=0)=1$.
For a finite-alphabet classical--quantum state this is equivalent to
$H(K\mid E,x=0)=0$: equality in $H(K\mid E,x=0)\geq0$ holds precisely when
the conditional Eve states have orthogonal supports. Eve may be
taken to hold a purification of the finite-dimensional Alice--Bob state, since any other extension arises from it by a channel on Eve's system, which cannot improve her guess. Her reduced state then has finite rank, so an optimal
finite-outcome guessing POVM exists on its support; when
$\pg(K\mid E,x=0)=1$, it guesses the key perfectly.

In each strategy conditioned on Eve's predicted key value, a relabeling of signs that preserves all parity and agreement conditions turns that key into $(+,+,+)$, and Bob's column postprocessing of the proof below then
preserves or increases the total score. After this change, the three input
pairs in Alice's key row contribute \(1/3\) to the score, and the remaining
six form a rectangle whose rows \(0,1\) are the square's rows \(1,2\), contributing at most \((2/3)\omega_Q(2,3)\).

\begin{proposition}[Exact reduction]\label{prop:reduction}
\[
\beta_0=\tfrac13+\tfrac23\omega_Q(2,3).
\]
\end{proposition}
\begin{proof}
Measure Eve's perfect guessing POVM. Conditional on each positive-probability
outcome, Alice's row 0 is one deterministic even string $k$.
The unconditioned Alice--Bob behavior is unchanged, so the score is an average
of such sector scores. Multiply Alice's entries in rows 0 and 1 in column
$j$ by $k_j$, and multiply Bob's matching two entries by $k_j$.
This preserves every parity and agreement predicate and makes row 0 equal
to $(+,+,+)$.

Fix one column. Let $L,R$ be Alice's coordinate observables, which are contractions, at this column's intersections with rows 1 and 2.
A Bob outcome with negative top sign is $(-,b,b)$, with POVM effect $N_b$.
Its contribution to the sum of the three intersection correlations of this column is
the expectation of $(-\Id+bL+bR)\otimes N_b$ in the conditional state.
Replace this outcome by a uniform choice of $(+,+,-)$ and $(+,-,+)$.
The increase is the expectation of
\[
[2\Id-b(L+R)]\otimes N_b\geq0.
\]
This is Bob-local stochastic postprocessing. All top entries are now positive,
so the top row always wins and the bottom rows form a rectangle.
Thus each sector's score is at most $1/3+(2/3)\omega_Q(2,3)$.
Conversely, any rectangle strategy embeds by adjoining a deterministic
positive top row and Bob's positive top entries, with Eve trivial.
Taking suprema proves the equality.
\end{proof}

\begin{corollary}[Sharp zero-entropy threshold]\label{cor:threshold}
For finite-dimensional tensor-product strategies with arbitrary POVMs and
quantum side information,
\begin{equation}\label{eq:threshold}
\beta_0=\frac{2+\sqrt{2/3}}3=0.9388321936425753\ldots.
\end{equation}
The value is attained with trivial Eve,
Alice rows $(\Id,\Id,\Id)$, $(D_1^0,D_2^0,D_3^0)$, $(-C_1^0,-C_2^0,-C_3^0)$,
and Bob columns $(\Id,B_j^0,-B_j^0)$ on $\Phi_4$.
\end{corollary}

The threshold concerns the score and Alice's raw-key entropy, with Eve's
system as side information. Only Alice's key is required to be perfectly predictable; Bob's column may remain random.
Applications to a key-distribution protocol additionally require its
reconciliation and acceptance conditions.

Equality at the threshold also determines the original strategies whose
key is perfectly predictable. The reduction above first changes Bob's
columns. To recover the original strategy, we show that equality forces
this change to vanish on each branch associated with a key value, and
then assemble the rectangle extractions by separate local isometries.
A single assembled pair of local isometries acts on the whole original
purification, preserving coherence between the key branches.

Write \(\mathcal A\coloneqq\{\alpha_1,\alpha_2,\alpha_3,\alpha_4\}\) for the
even sign strings in \eqref{eq:alpha}. For \(s=\alpha_a\), abbreviate
\(P_s^0\coloneqq P_a^0\) and \(Q_s^0\coloneqq Q_a^0\), and write \(ks\) for
coordinatewise multiplication of sign strings. Let \(A_{s|x}\) be
Alice's original row effects, with \(x=0,1,2\) and \(s\in\mathcal A\).
Let \(G_{c|j}\) be Bob's original column effects, where
\(c=(c_0,c_1,c_2)\in\{\pm1\}^3\) has product \(-1\).

\begin{theorem}[Classification at the threshold]
\label{thm:threshold-classification}
Consider a finite-dimensional tensor-product strategy for the uniform
magic-square game, with arbitrary POVMs and a purification
\(\psi_{ABE}\). Suppose its score \(w\) and key satisfy
\[
w=\beta_0=\frac{2+\sqrt{2/3}}3,
\qquad \pg(K\mid E,x=0)=1.
\]
There is a nonempty set \(\mathcal I\subseteq\mathcal A\) of
positive-probability key values, orthogonal decompositions of the
occupied local supports
\[
\supp\rho_A=\bigoplus_{k\in\mathcal I}\mathcal H_{A,k},
\qquad
\supp\rho_B=\bigoplus_{k\in\mathcal I}\mathcal H_{B,k},
\]
and local support identifications
\[
W_{A,k}:\mathcal H_{A,k}\longrightarrow
\mathbb C^4\otimes\mathcal K_{A,k},\qquad
W_{B,k}:\mathcal H_{B,k}\longrightarrow
\mathbb C^4\otimes\mathcal K_{B,k}.
\]
All original measurement effects preserve these subspaces.
Under these identifications their restrictions are
\begin{align*}
A_{s|0}|_k&=\delta_{s,k}\Id,&
A_{s|1}|_k&=P_{ks}^0\otimes\Id,&
A_{s|2}|_k&=Q_s^0\otimes\Id,\\
G_{c|j}|_k&=
\begin{cases}
\displaystyle\frac{\Id_4+k_jc_1 B_j^0}{2}\otimes\Id,
&c_0=k_j,\\[2pt]
0,&c_0=-k_j.
\end{cases}
\end{align*}
In the nonzero case the legal column parity already implies
\(c_2=-k_jc_1\).

Put \(\mathcal K_A\coloneqq\bigoplus_k\mathcal K_{A,k}\) and
\(\mathcal K_B\coloneqq\bigoplus_k\mathcal K_{B,k}\).
After adjoining unused ancillary dimensions if necessary, the direct sums
of the support identifications, followed by the natural identifications
with \(\mathbb C^4\otimes\mathcal K_A\) and \(\mathbb C^4\otimes\mathcal K_B\),
extend to Alice-local and Bob-local isometries \(V_A,V_B\) of the original spaces into \(\mathbb C^4_{A_0}\otimes\mathcal H_{A'}\) and \(\mathbb C^4_{B_0}\otimes\mathcal H_{B'}\), where \(\mathcal H_{A'}\supseteq\mathcal K_A\) and \(\mathcal H_{B'}\supseteq\mathcal K_B\). They satisfy
\[
(V_A\otimes V_B\otimes\Id_E)\psi
=\ket{\Phi_4}_{A_0B_0}\otimes\ket\eta_{A'B'E}.
\]
Moreover, there are pairwise orthogonal subspaces
\(\mathcal E_k\subseteq\mathcal H_E\) and vectors
\(\eta_k\in\mathcal K_{A,k}\otimes\mathcal K_{B,k}
\otimes\mathcal E_k\) such that
\[
\eta=\sum_{k\in\mathcal I}\sqrt{p_k}\,\eta_k,
\qquad \|\eta_k\|=1,\quad p_k>0,\quad \sum_kp_k=1.
\]
Here \(p_k\) is the probability of the key value \(k\).
Thus the extraction acts on Alice and Bob alone and preserves the
coherence of the original purification. Each subspace
\(\mathcal H_{A,k}\) or \(\mathcal H_{B,k}\), and hence each total
occupied local support, has dimension divisible by four.

Conversely, the displayed state and measurement restrictions, including
the orthogonality of the Eve subspaces, define a strategy with
\(w=\beta_0\) and perfect Eve prediction. Measurements on unused local
dimensions are unrestricted subject to being POVMs.
\end{theorem}

\begin{proof}
First resolve the key on Eve's original system. The subnormalized
conditional Eve states are
\[
\tau_k\coloneqq\Tr_{AB}\bigl[
(\sqrt{A_{k|0}}\otimes\Id_{BE})\ket\psi\bra\psi
(\sqrt{A_{k|0}}\otimes\Id_{BE})\bigr].
\]
They sum to \(\rho_E\). This state has finite rank because Alice and
Bob are finite-dimensional, so an optimal finite-outcome guessing POVM
exists on its support. If \(\{M_k\}\) is a perfect guessing POVM, positivity
gives \(M_k=\Id\) on \(\supp\tau_k\) and \(M_k=0\) on
\(\supp\tau_l\) for \(l\ne k\). The subspaces
\(\mathcal E_k\coloneqq\supp\tau_k\) are therefore pairwise orthogonal.
Let \(\Pi_k\) be their orthogonal projections. Their sum is the identity
on \(\supp\rho_E\).

Put \(p_k\coloneqq\Tr\tau_k\), omit the values with \(p_k=0\), and define normalized vectors
\[
\psi_k\coloneqq p_k^{-1/2}(\Id_{AB}\otimes\Pi_k)\psi.
\]
For \(s\ne k\),
\[
\| (\sqrt{A_{s|0}}\otimes\Pi_k)\psi\|^2
=\Tr(\Pi_k\tau_s)=0.
\]
Consequently \(A_{s|0}\psi_k=0\) for \(s\ne k\), and
\(A_{k|0}\psi_k=\psi_k\). As in the proof of Lemma~\ref{lem:support}, an operator on one party
that annihilates a pure state vanishes on the support of that party's reduced state. Thus on
\(\mathcal H_{A,k}\coloneqq\supp\Tr_{BE}\ket{\psi_k}\bra{\psi_k}\),
the original key POVM is deterministically \(k\). These Alice supports
are pairwise orthogonal: \(A_{k|0}\) has eigenvalue one on
\(\mathcal H_{A,k}\) and zero on \(\mathcal H_{A,l}\), \(l\ne k\).

The decomposition
\(\psi=\sum_k\sqrt{p_k}\psi_k\) is an identity of vectors; no
measurement has been performed. Every game effect acts trivially on Eve,
so orthogonality of the \(\mathcal E_k\) gives
\(w=\sum_kp_kw_k\), where \(w_k\) is the score of the branch \(\psi_k\). Proposition~\ref{prop:reduction} gives
\(w_k\leq\beta_0\). Since \(p_k>0\) and the average equals
\(\beta_0\), every \(w_k\) equals \(\beta_0\).

Fix one such \(k\). We repeat the relabeling and column postprocessing of the
proof of Proposition~\ref{prop:reduction}, now with explicit effects.
For this proof only, multiply coordinates in Alice's
rows 0 and 1 by \(k\), and Bob's entries in rows 0 and 1 of column \(j\)
by \(k_j\). This preserves all parities and agreement predicates, and
makes the key \((+,+,+)\). Let \(L_j,R_j\) be the \(j\)th coordinate
contractions of Alice's relabeled rows 1 and 2. Bob's four legal
outcomes in column \(j\) have the effects
\[
F_b^j\coloneqq G_{(k_j,k_jb,-b)|j},\qquad
N_b^j\coloneqq G_{(-k_j,k_jb,b)|j},\qquad b=\pm1.
\]
Put \(N_j\coloneqq N_+^j+N_-^j\). Replacing each outcome with effect
\(N_b^j\) by a uniform choice of the outcomes with effects \(F_\pm^j\)
gives binary effects \(H_b^j\coloneqq F_b^j+\tfrac12N_j\), with
\(H_+^j+H_-^j=\Id\). The full score is \(1/2\) plus
\(1/18\) times the sum of the nine intersection correlations, and this replacement increases it by
\[
\Delta_k\coloneqq\frac1{18}\sum_{j,b}
\bigl\langle[2\Id-b(L_j+R_j)]\otimes N_b^j\bigr\rangle_{\psi_k}\geq0,
\]
since \(2\Id-b(L_j+R_j)\geq0\).
The postprocessed strategy has a top row that always wins and a rectangle on
its last two rows, so its score is at most
\(1/3+(2/3)\omega_Q(2,3)=\beta_0=w_k\). Hence \(\Delta_k=0\) and the
rectangle is optimal.

Apply Theorem~\ref{thm:rigidity} to this rectangle with its original
branch purification \(\psi_k\), keeping Eve as a separate system. It
supplies \(W_{A,k},W_{B,k}\) and the state \(\Phi_4\otimes\eta_k\); write
\(\rho_{B,k}\coloneqq\Tr_{AE}\ket{\psi_k}\bra{\psi_k}\) and
\(\mathcal H_{B,k}\coloneqq\supp\rho_{B,k}\). On the invariant subspace
\(\mathcal H_{A,k}\), \(L_j=D_j^0\otimes\Id\) and \(R_j=-C_j^0\otimes\Id\),
so \((L_j+R_j)^2=\frac43\Id\) there by \eqref{eq:construction}. As
\(\psi_k\) has Alice support \(\mathcal H_{A,k}\),
\[
\bigl\langle[2\Id-b(L_j+R_j)]\otimes N_b^j\bigr\rangle_{\psi_k}
\geq\Bigl(2-\frac2{\sqrt3}\Bigr)\Tr(\rho_{B,k}N_b^j)\geq0.
\]
Since \(\Delta_k=0\), every \(\Tr(\rho_{B,k}N_b^j)\) vanishes; thus
\((N_b^j)^{1/2}\rho_{B,k}^{1/2}=0\), and \(N_b^j\) and \(N_j\) vanish on
\(\mathcal H_{B,k}\). Hence \(F_b^j=H_b^j\) on \(\mathcal H_{B,k}\) as
maps, not merely in expectation. By Theorem~\ref{thm:rigidity}, Alice's effects in rows
1 and 2 preserve \(\mathcal H_{A,k}\), the \(H_b^j\) preserve
\(\mathcal H_{B,k}\), and all are the reference effects there. So all
original column effects preserve this support, the \(F_b^j\) agree there
with the reference rectangle effects, and the \(N_b^j\) vanish there.
Undoing the relabeling yields exactly the displayed row and column
formulas.

For distinct keys \(k,l\), choose a coordinate \(j\) with
\(k_j\ne l_j\). Bob's original top-sign marginal
\[
T_{k_j|j}\coloneqq\sum_{c:c_0=k_j}G_{c|j}
\]
is the identity on \(\mathcal H_{B,k}\) and zero on
\(\mathcal H_{B,l}\), by the preceding paragraph applied separately
to both branches. Its eigenspaces for one and zero are orthogonal,
proving \(\mathcal H_{B,k}\perp\mathcal H_{B,l}\).

Tracing Eve kills the cross terms between different \(\psi_k\).
Consequently each total local reduced state is the positive weighted sum
of its branch reduced states, and its support is exactly the corresponding
direct sum stated in the theorem. All row and column effects preserve
each summand. The direct sums of \(W_{A,k}\) and \(W_{B,k}\) are
therefore locally implementable isometries on these supports. They send
the original coherent vector to
\[
\sum_k\sqrt{p_k}\,\Phi_4\otimes\eta_k
=\Phi_4\otimes\sum_k\sqrt{p_k}\eta_k.
\]
Here the natural embeddings of the ancillary summands are understood.
Because the local maps leave Eve
unchanged, \(\eta_k\) is supported in \(\mathcal E_k\).
Orthogonality and normalization of the branches show that the ancillary
sum is normalized. Enlarge the ancillary spaces if necessary and map
unused local orthogonal complements into additional target dimensions
to extend the support identifications to isometries of the original local spaces.
All assertions on occupied supports remain unchanged.

Finally, in the asserted normal form Alice's key is \(k\) on the local
summands indexed by \(k\), and Eve recovers it by measuring the mutually
orthogonal \(\mathcal E_k\). Undoing the same parity-preserving signs
in each corresponding branch leaves a deterministic winning top row and
the optimal reference rectangle, so every branch and their average have
score \(\beta_0\).
The dimension assertions follow from
\(\mathcal H_{A,k}\simeq\mathbb C^4\otimes\mathcal K_{A,k}\)
and its Bob counterpart.
\end{proof}

A coherent direct sum of local isometries also appears in the extraction
argument of \cite[Corollary~3.6, arXiv v2]{MNP}. Here the equality conditions
for the column postprocessing establish the required orthogonal decompositions
of the local supports.

\section{Guessing probability above the threshold}\label{sec:guessing}
Perfect predictability is lost as soon as the score exceeds $\beta_0$.
To quantify this loss, replace Alice's key row by a fixed even string.
Any other even string differs from it in exactly two coordinates, so the
score loss is at most $2/9$ times the probability that the replacement
changes the output. Averaging this comparison with Eve's guessing POVM
gives a linear bound for general quantum side information.

\begin{theorem}[Linear guessing bound]\label{thm:linear}
Let a finite-dimensional magic-square strategy with arbitrary local POVMs
and quantum side information have score $w$. For Alice's complete first
row $K$,
\begin{equation}\label{eq:linear}
w\leq\beta_0+\frac29\bigl(1-\pg(K\mid E,x=0)\bigr).
\end{equation}
In particular, if $\beta_0\leq w\leq1$, then
\begin{align}
\pg(K\mid E,x=0)&\leq1-\frac92(w-\beta_0),\label{eq:pguess}\\
H_{\min}(K\mid E,x=0)&\geq
-\log_2\!\left(1-\frac92(w-\beta_0)\right).\label{eq:minentropy}
\end{align}
Here $H_{\min}=-\log_2\pg$ is the unsmoothed conditional min-entropy
of the one-round classical key \cite[Theorem~1]{KRS}.
\end{theorem}
\begin{proof}
Write Alice's row effects as $A_{s|x}$ and Bob's column effects as
$G_{c|j}$, where $s$ ranges over even and $c$ over odd strings. The
winning effect is
\[
\mathcal W\coloneqq\frac19\sum_{x,j,s,c}
\mathbf1[s_j=c_x]A_{s|x}\otimes G_{c|j},
\qquad w=\Tr(\mathcal W\rho_{AB}).
\]
For a fixed even string $k$, let $\mathcal W^{(k)}$ be the winning effect
after replacing Alice's row-$0$ measurement by the deterministic output
$k$, keeping the state and all other measurements. When
$\sum_sA_{s|0}=\Id$ is inserted into its row-$0$ terms, the coefficient of
$A_{s|0}\otimes G_{c|j}$ in $\mathcal W-\mathcal W^{(k)}$ is
$\frac19(\mathbf1[s_j=c_0]-\mathbf1[k_j=c_0])\leq\frac19\mathbf1[s_j\ne k_j]$.
Positivity of the effects and completeness of Bob's measurements give
\begin{equation}\label{eq:replacement}
\mathcal W-\mathcal W^{(k)}
\leq\frac19\sum_s d_H(s,k)A_{s|0}\otimes\Id
=\frac29(\Id-A_{k|0})\otimes\Id,
\end{equation}
with the Hamming distance $d_H$, because distinct even strings of
length three differ in two coordinates.
The replaced key is deterministic for every bipartite state, so the
definition of $\beta_0$ gives $\mathcal W^{(k)}\leq\beta_0\Id$.

For any Eve guessing POVM $\{M_k\}$, combine the two bounds, tensor with
$M_k\geq0$ and sum over $k$. Since $\sum_kM_k=\Id_E$, this yields
\[
\mathcal W\otimes\Id_E
\leq\beta_0\Id+
\frac29\Bigl(\Id-\sum_kA_{k|0}\otimes\Id_B\otimes M_k\Bigr).
\]
Taking the expectation in the original state gives
$w\leq\beta_0+(2/9)(1-p_M)$, where
$p_M\coloneqq\sum_k\Tr[(A_{k|0}\otimes\Id_B\otimes M_k)\rho_{ABE}]$ is the
success probability of this guess.
Taking the supremum of $p_M$ proves \eqref{eq:linear} without
requiring an optimal guessing measurement. Rearranging and applying
$-\log_2$ proves \eqref{eq:pguess} and \eqref{eq:minentropy}.
The logarithm's argument is positive on the stated interval: at $w=1$
it equals $(\sqrt6-1)/2$.
\end{proof}

The coefficient in \eqref{eq:pguess} gives a lower bound on the loss of
predictability. A matching linear order follows by mixing the threshold
strategy with a perfect magic-square strategy. For completeness, the
standard perfect strategy \cite{Mermin,Peres}, written in the present sign
convention, uses the Pauli matrices
\[
\mathsf X\coloneqq\begin{pmatrix}0&1\\1&0\end{pmatrix},\qquad
\mathsf Z\coloneqq\begin{pmatrix}1&0\\0&-1\end{pmatrix}
\]
and the following array of observables on $\mathbb C^2\otimes\mathbb C^2$:
\begin{equation}\label{eq:perfect-square}
\begin{pmatrix}
\mathsf Z\otimes\Id&\Id\otimes\mathsf Z&\mathsf Z\otimes\mathsf Z\\
\Id\otimes\mathsf X&\mathsf X\otimes\Id&\mathsf X\otimes\mathsf X\\
-\mathsf Z\otimes\mathsf X&-\mathsf X\otimes\mathsf Z&
-(\mathsf X\mathsf Z)\otimes(\mathsf X\mathsf Z)
\end{pmatrix}.
\end{equation}
Each row is a commuting reflection triple with product $+\Id$, and each
column is a commuting reflection triple with product $-\Id$. Alice
measures rows, Bob measures the transposed columns, and they share
$\Phi_4$. Every intersection correlation is one. The first row has four
rank-one joint projections, so its key is uniform and independent of a
trivial Eve.

Let $g(w)$ be the supremum of $\pg(K\mid E,x=0)$ over the strategies of
Theorem~\ref{thm:linear} with score at least $w$. Also let $h_{\min}(w)$ be the infimum of their
conditional min-entropies; equivalently, $h_{\min}(w)=-\log_2g(w)$.
These functions describe the best guessing probability and the least
entropy compatible with a score constraint.

For the same strategy class, define
\[
h_{\mathrm{vN}}(w)\coloneqq\inf_{\text{score}\geq w}H(K\mid E,x=0),
\]
where $H$ is conditional von Neumann entropy, in bits.

\begin{corollary}[Linear onset of randomness]\label{cor:gap}
For $0\leq\delta\leq1-\beta_0$,
\begin{equation}\label{eq:gap}
\frac92\delta\leq1-g(\beta_0+\delta)
\leq\frac{3\delta}{4(1-\beta_0)}.
\end{equation}
For conditional von Neumann entropy,
\begin{equation}\label{eq:vn-gap}
-\log_2\!\left(1-\frac92\delta\right)
\leq h_{\mathrm{vN}}(\beta_0+\delta)
\leq\frac{2\delta}{1-\beta_0}.
\end{equation}
Consequently $1-g(\beta_0+\delta)$, $h_{\min}(\beta_0+\delta)$ and
$h_{\mathrm{vN}}(\beta_0+\delta)$ are all $\Theta(\delta)$ as
$\delta\downarrow0$.
\end{corollary}
\begin{proof}
The first inequality follows from Theorem~\ref{thm:linear}. Mix the
deterministic-key threshold strategy of Corollary~\ref{cor:threshold}
with the perfect strategy
\eqref{eq:perfect-square}, giving the latter weight $t$. Store the branch
in orthogonal local flags and give its value to Eve. The score and guessing
probability are
\[
w_t=(1-t)\beta_0+t,\qquad p_t=1-\frac34t.
\]
Choosing $t\coloneqq\delta/(1-\beta_0)$ proves the second inequality. Taking
$-\log_2$ gives
\[
-\log_2(1-\tfrac92\delta)
\leq h_{\min}(\beta_0+\delta)
\leq-\log_2\!\left(1-\frac{3\delta}{4(1-\beta_0)}\right),
\]
which proves the min-entropy order as well.
For a classical key with quantum side information,
$H(K\mid E)\geq H_{\min}(K\mid E)$ \cite{KRS}, so the same lower
bound holds for $h_{\mathrm{vN}}$. In the flagged mixture above, the threshold
branch has conditional entropy zero and the perfect branch has conditional
entropy two. Since Eve knows the flag, the mixture has conditional entropy
$2t$. This proves the upper bound in \eqref{eq:vn-gap}.
Both bounds are linear to first order at zero, proving the final assertion.
\end{proof}

Every strategy of Theorem~\ref{thm:linear} with score $w$ satisfies $H(K\mid E,x=0)\geq f(w)$, where
\[
f(w)\coloneqq\max\left\{0,-\log_2\!\left[1-\frac92(w-\beta_0)\right]\right\},
\qquad 0\leq w\leq1.
\]
The logarithm is defined on this interval. The function $-\log_2[1-\frac92(w-\beta_0)]$ is convex and
nondecreasing, and taking the maximum with zero preserves both properties.
Above $\beta_0$ this is the lower bound in \eqref{eq:vn-gap}; below
$\beta_0$, nonnegativity of classical--quantum conditional entropy supplies it.

The bound $f$ also applies to the magic-square games of \cite{CRS}.
In the convention of \cite[Secs.~4.1 and~4.5]{CRS}, with uniform independent inputs and Alice's key input $x=0$ as here, Alice and Bob may each return
any triple of signs and win only when Alice's parity is even, Bob's is odd and the agreement condition holds. Replace the last sign of each answer so that Alice's
parity is even and Bob's is odd. This local postprocessing leaves every
winning answer unchanged, so the resulting score $w'$ satisfies $w'\geq w$. The postprocessed strategy is a strategy of Theorem~\ref{thm:linear}, with key
$\widetilde K$, and $\widetilde K$ is in bijection with Alice's original first two signs $T$.
Consequently,
\[
H(T\mid E,x=0)=H(\widetilde K\mid E,x=0)
\geq f(w')\geq f(w),
\]
and the full original row $K$ has conditional entropy at least that of $T$.

In the parity-extended magic-square games of \cite[Sec.~4.5, Eq.~(9)]{CRS},
Alice and Bob play with $N-2\geq0$ additional parties, who receive no input and each output two bits; let $p_1$ and $p_2$ be the parities of their first and second bits. With uniform independent inputs, the players win when Alice's row has product $+1$, Bob's column has product $-1$, and Alice's $j$th sign times Bob's sign in row $x$ equals $(-1)^{\sigma(x,j)}$. The mask $\sigma$ vanishes in row $x=0$, Alice's key input, and equals $p_1$, $p_2$ and $p_1\oplus p_2$ in the three columns of rows $1$ and $2$ \cite[Eq.~(8), Lemma~4.6(ii)]{CRS}. The proof of
\cite[Lemma~6.6, Eqs.~(21)--(22)]{CRS} conditions on the additional
parties' outputs, reduces each branch to a two-player strategy, and applies
Jensen's inequality to a convex two-player bound. The bound $f$ applies to each branch: there the mask is fixed and equal in rows $1$ and $2$, so multiplying Bob's signs in these two rows of column $j$ by $(-1)^{\sigma(1,j)}$ keeps his parity, leaves row $0$ unchanged and gives a strategy for the two-player game of the previous paragraph with the same winning probability. The $N$-party winning probability $w_N$ is the average of these branch probabilities \cite[Eq.~(25)]{CRS}, because the additional parties receive no input. Since $f$ is convex, this
argument gives, for every number $N\geq2$ of parties, Alice and Bob included, and every
finite-dimensional local strategy with quantum side information $E$, which may also include the additional parties' outputs,
\[
H(K\mid E,x=0)\geq H(T\mid E,x=0)\geq f(w_N),
\]
where, in a single round, $w_N$ is the winning probability of the $N$-party game, $x=0$ is Alice's key input, and $K$ and $T$ are her full output row and its first two signs at that input.

The numerical bound $33.4\,w-31.47$ of \cite[Lemma~6.7]{CRS}, which also holds for the parity-extended games by \cite[Lemma~6.6]{CRS}, bounds the same entropy $H(K\mid E,x=0)$ in terms of the same two-player score $w$. It is negative for $w<31.47/33.4\approx0.9422$, while $f$ is positive above $\beta_0\approx0.9388$, and it exceeds $f$ above $w\approx0.943$, reaching $1.93$ bits at $w=1$ against $f(1)\approx0.46$. Figure~\ref{fig:entropy} plots both bounds.

\begin{figure}[t]
\centering
\BeginAccSupp{method=escape,Alt={Two line plots compare entropy lower bounds. The solid bound turns positive at a lower score; the dashed bound is larger above about 0.943. See the caption for values.},ActualText={Two line plots compare entropy lower bounds. The solid bound turns positive at a lower score; the dashed bound is larger above about 0.943. See the caption for values.}}%
\begin{tikzpicture}
\begin{axis}[width=0.47\textwidth,height=5cm,xmin=0.938,xmax=1,ymin=0,ymax=2.2,
  xlabel={score $w$},ylabel={bits},xtick={0.94,0.96,0.98,1},scaled x ticks=false,
  xticklabel style={/pgf/number format/fixed,/pgf/number format/fixed zerofill,/pgf/number format/precision=2},
  legend pos=north west,legend cell align=left,legend style={font=\footnotesize,draw=none,fill=white,fill opacity=0.8,text opacity=1}]
\addplot[black,thick] coordinates {(0.93883,0.00000) (0.94006,0.00796) (0.94128,0.01597) (0.94250,0.02403) (0.94373,0.03212) (0.94495,0.04027) (0.94617,0.04846) (0.94740,0.05669) (0.94862,0.06498) (0.94984,0.07331) (0.95107,0.08169) (0.95229,0.09012) (0.95351,0.09860) (0.95474,0.10713) (0.95596,0.11571) (0.95718,0.12434) (0.95841,0.13302) (0.95963,0.14176) (0.96085,0.15055) (0.96208,0.15939) (0.96330,0.16829) (0.96452,0.17724) (0.96575,0.18625) (0.96697,0.19531) (0.96819,0.20444) (0.96942,0.21362) (0.97064,0.22286) (0.97186,0.23216) (0.97309,0.24151) (0.97431,0.25093) (0.97553,0.26042) (0.97676,0.26996) (0.97798,0.27957) (0.97920,0.28924) (0.98043,0.29898) (0.98165,0.30878) (0.98287,0.31866) (0.98410,0.32860) (0.98532,0.33860) (0.98654,0.34868) (0.98777,0.35883) (0.98899,0.36905) (0.99021,0.37935) (0.99144,0.38971) (0.99266,0.40016) (0.99388,0.41068) (0.99511,0.42127) (0.99633,0.43195) (0.99755,0.44270) (0.99878,0.45354) (1.00000,0.46445)};
\addlegendentry{$f(w)$}
\addplot[black,thick,dashed] coordinates {(0.93800,0.00000) (0.94222,0.00000) (1.00000,1.93000)};
\addlegendentry{$\max\{0,\,33.4\,w-31.47\}$}
\end{axis}
\end{tikzpicture}\hfill
\begin{tikzpicture}
\begin{axis}[width=0.47\textwidth,height=5cm,xmin=0.938,xmax=0.946,ymin=0,ymax=0.13,
  xlabel={score $w$},ylabel={bits},xtick={0.938,0.940,0.942,0.944,0.946},scaled x ticks=false,scaled y ticks=false,
  xticklabel style={/pgf/number format/fixed,/pgf/number format/fixed zerofill,/pgf/number format/precision=3},
  yticklabel style={/pgf/number format/fixed,/pgf/number format/precision=2}]
\addplot[black,thick] coordinates {(0.93800,0.00000) (0.93883,0.00000) (0.93901,0.00116) (0.93919,0.00233) (0.93937,0.00349) (0.93955,0.00466) (0.93973,0.00583) (0.93991,0.00700) (0.94009,0.00817) (0.94027,0.00934) (0.94044,0.01051) (0.94062,0.01168) (0.94080,0.01285) (0.94098,0.01403) (0.94116,0.01520) (0.94134,0.01638) (0.94152,0.01756) (0.94170,0.01873) (0.94188,0.01991) (0.94206,0.02109) (0.94224,0.02227) (0.94242,0.02346) (0.94260,0.02464) (0.94277,0.02582) (0.94295,0.02701) (0.94313,0.02819) (0.94331,0.02938) (0.94349,0.03057) (0.94367,0.03176) (0.94385,0.03295) (0.94403,0.03414) (0.94421,0.03533) (0.94439,0.03652) (0.94457,0.03772) (0.94475,0.03891) (0.94492,0.04011) (0.94510,0.04130) (0.94528,0.04250) (0.94546,0.04370) (0.94564,0.04490) (0.94582,0.04610) (0.94600,0.04730)};
\addplot[black,thick,dashed] coordinates {(0.93800,0.00000) (0.94222,0.00000) (0.94600,0.12640)};
\draw[dotted] (axis cs:0.93883,0) -- (axis cs:0.93883,0.13) node[pos=0.92,right,font=\footnotesize] {$\beta_0$};
\end{axis}
\end{tikzpicture}%
\EndAccSupp{}
\caption{Lower bounds on the conditional von Neumann entropy $H(K\mid E,x=0)$ of Alice's key, in bits, as functions of the magic-square score $w$: the bound $f$ of this section (solid) and the numerical bound $33.4\,w-31.47$ of \cite[Lemma~6.7]{CRS}, replaced by zero where it is negative (dashed). Left: $0.938\leq w\leq1$. Right: the window $0.938\leq w\leq0.946$ near the threshold $\beta_0\approx0.9388$ (dotted). The bound $f$ is positive for every $w>\beta_0$; the numerical bound is positive only for $w>31.47/33.4\approx0.9422$, and it exceeds $f$ for $w$ above about $0.943$.}\label{fig:entropy}
\end{figure}
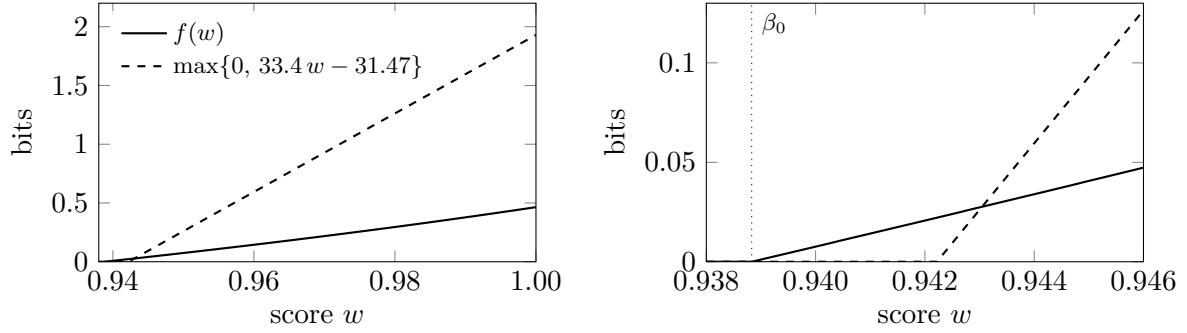

Appendix~\ref{app:projection} contains a general projection argument giving a quadratic bound; the
linear bound uses the additional distance structure of the four legal
key outputs.
\section{Optimal codes hiding the full input parity}\label{sec:rac}
The two row measurements can also be viewed as remote preparations of two
ensembles. Their common average state hides the row label, which is the
parity of the encoded sign string. This correspondence transfers both the
rectangle optimum and its rigidity to a prepare-and-measure random access
code. We first establish the equality of scores in the two models and then
classify the original encoding states.
Consider eight equiprobable normalized states $\rho_s$, indexed by
$s\in\{\pm1\}^3$, and binary decoding measurements. The requested coordinate
is independent of the input string and uniform on $\{1,2,3\}$.
Success is averaged over the eight inputs and three coordinate requests.
Require only
\begin{equation}\label{eq:privacy}
\frac14\sum_{\prod s_j=+1}\rho_s
=\frac14\sum_{\prod s_j=-1}\rho_s.
\end{equation}
The common encoding space has arbitrary finite dimension. Two-bit parities
may be revealed.

Parity-oblivious encodings were introduced in \cite{SBKTP};
the present model is the three-bit instance of the odd-parity-oblivious
model distinguished by Chailloux, Kerenidis, Kundu, and Sikora \cite{CKKS}, which hides the parity of every subset of odd size at least three.
They leave the optimal odd-parity-oblivious bounds as an open problem
\cite[Sec.~1.1, arXiv v1]{CKKS}; for even-parity-oblivious encodings they
prove an analogous equivalence with the INDEX game \cite[Theorem~2]{CKKS}.
The result below determines the optimal three-bit
average success probability for arbitrary finite encoding dimension.
The three-bit prepare-and-measure model of Singh, Sasmal, Nautiyal, and Pan
\cite[Sec.~IV.2, arXiv v1]{Singh} instead hides all three two-bit parities
as well as the full parity. Its feasible encodings therefore obey stronger
privacy constraints than the codes considered here.

Roy and Pan consider the same full-parity constraint on a
four-dimensional encoding space. In their arXiv v1, the paragraph following
Eq.~(25) imposes, for one coordinate $j$, equality of $\sum_ss_j\rho_s$ over the even and over the odd strings $s$, such as their $A_{3,1}=A_{3,2}$ for $j=1$, stated to be without loss of generality. The
resulting calculation gives $1/2+(1+\sqrt2)/6\approx0.902369$, which they state to be optimal
\cite[Sec.~IV, Eqs.~(17), (25)--(27), arXiv v1]{RP}.
The known four-dimensional code, displayed below, satisfies \eqref{eq:privacy}
and attains $1/2+1/\sqrt6\approx0.908248$ without such an equality: for it the two sums are $D_j^0\neq C_j^0$.
Since their success probability is also the average over the eight inputs
and three coordinate requests, the displayed code achieves a strictly
larger success probability in the same four-dimensional problem.

The larger value can be checked directly with the reference projections:
\[
\rho_{\alpha_a}\coloneqq P_a^0,\qquad \rho_{-\alpha_a}\coloneqq Q_a^0,\qquad
F_{b|j}\coloneqq\frac{\Id+bX_j^0/q}{2}.
\]
Both parity averages equal $\Id_4/4$. Taking diagonal entries of the anticommutator in \eqref{eq:construction} in the two bases gives
success probability $(1+\sqrt{2/3})/2$ for all eight inputs and all three
coordinate requests.

\begin{proposition}[Equality of optimal scores]\label{prop:RAC}
The optimal average success probability under \eqref{eq:privacy} equals
$\omega_Q(2,3)$ and hence equals \eqref{eq:value}.
\end{proposition}
\begin{proof}
For a rectangle strategy, use Alice's first row as an even sign string and
negate the second row to obtain an odd sign string. Let $\sigma_s$ be Bob's
subnormalized conditional state. No signalling gives the same sum $\rho_B$
within both parity classes.
To make the eight conditional states equiprobable, append a shared
uniform even sign string $g$ and retain $g$ in an orthogonal Bob register. Multiply each Alice sign and Bob's coordinate guess by $g_j$.
Products such as $gs$ below are coordinatewise products of sign strings.
The new subnormalized states are
\[
\sigma'_s\coloneqq\tfrac14\sum_{g\text{ even}}\ket g\bra g\otimes\sigma_{gs}.
\]
They all have trace $1/4$, preserve the score, and their sums in the two
parity classes agree. Thus $\rho'_s\coloneqq4\sigma'_s$ is an admissible code.

Conversely, let $\tau$ be the common parity average in \eqref{eq:privacy}.
On its support use the canonical purification $(\Id\otimes\sqrt\tau)\sum_t\ket{t,t}$ of $\tau$ in a chosen basis, with Alice on the first factor,
and Alice effects
\[
E_s\coloneqq\tfrac14\bigl(\tau^{-1/2}\rho_s\tau^{-1/2}\bigr)^T.
\]
Within each parity class the effects sum to identity. They steer Bob to
$\rho_s/4$, since $\sqrt\tau E_s^T\sqrt\tau=\rho_s/4$.
Use the code's decoders as Bob's measurements and negate Alice's odd row
back to even parity. This realizes the rectangle with the same score.
All inverses are restricted to $\supp\tau$; every $\rho_s$ is supported
there because their positive average is $\tau$.
\end{proof}

The construction of a code from a rectangle strategy enlarges Bob's dimension by a factor of four.
Thus a bound for codes on two qubits does not yield the unrestricted upper
bound through this equivalence. Suzuki's two orthonormal bases, one for
the even input strings and one for the odd input strings \cite{Suzuki},
satisfy \eqref{eq:privacy} immediately, so the converse construction
above turns them into a four-dimensional optimal rectangle strategy by
remote preparation. With the bit $0$ read as the sign $+1$, conjugation by $D_2^0$ maps the states and decoders of Suzuki's $(n,n-1)$ code for $n=3$ \cite[Secs.~III.A--III.B]{Suzuki} exactly to the code $\rho_{\pm\alpha_a}$, $F_{b|j}$ displayed before Proposition~\ref{prop:RAC}.

The converse construction preserves the code's original encoding
system. Applying rectangle rigidity in that direction therefore determines
all optimal encoding states with one common isometry. The support
qualification below separates the part of the Hilbert space used by the
code from unused dimensions.

\begin{theorem}[Classification of optimal codes]\label{thm:RACclass}
Let $\rho_s$, $s\in\{\pm1\}^3$, be normalized states on a common
finite-dimensional space $\mathcal H$ satisfying \eqref{eq:privacy}, and
let $F_{b|j}$, $b\in\{\pm1\}$, be binary decoding POVMs. Suppose
\[
\frac1{24}\sum_{s,j}\Tr(\rho_sF_{s_j|j})
=\frac{1+\sqrt{2/3}}2.
\]
There are an ancillary space $\mathcal K$, a state $\sigma$ on it, and
one isometry $V:\mathcal H\longrightarrow\mathbb C^4\otimes\mathcal K$
such that, for every $a\in\{1,2,3,4\}$,
\begin{equation}\label{eq:RACnormal}
V\rho_{\alpha_a}V^*=P_a^0\otimes\sigma,
\qquad
V\rho_{-\alpha_a}V^*=Q_a^0\otimes\sigma.
\end{equation}
If $\tau$ is the common parity average in \eqref{eq:privacy}, then
$V(\supp\tau)=\mathbb C^4\otimes\supp\sigma$. On this subspace,
\begin{equation}\label{eq:RACdecoder}
VF_{b|j}V^*=\frac{\Id_4+bB_j^0}{2}\otimes\Id.
\end{equation}
Conversely, these states and decoder restrictions attain the optimum.
In particular, $\operatorname{rank}\tau=4\operatorname{rank}\sigma$,
and every encoding state has the same rank as $\sigma$.
\end{theorem}
\begin{proof}
Every $\rho_s$ is supported on $\supp\tau$, since $\rho_s\leq4\tau$. On this support
choose a basis and the canonical purification
$\ket{\psi_\tau}\coloneqq(\Id\otimes\sqrt\tau)\sum_t\ket{t,t}$.
Use the effects $E_s$ from the proof of Proposition~\ref{prop:RAC} on the
purifying system and the
original decoders on the encoding system. The conditional encoding state
is $\rho_s/4$, each outcome has probability $1/4$, and the rectangle
score equals the code score. The strategy is therefore optimal.

Theorem~\ref{thm:rigidity} supplies local support identifications $W_A,W_B$
such that
\[
(W_A\otimes W_B)\ket{\psi_\tau}
=\ket{\Phi_4}\otimes\ket\eta_{A'B'}.
\]
Alice's effects $E_{\alpha_a},E_{-\alpha_a}$ become
$P_a^0\otimes\Id,Q_a^0\otimes\Id$, and Bob's decoding observables
become $B_j^0\otimes\Id$. Here Bob's system is precisely the original
encoding support. Put $\sigma\coloneqq\Tr_{A'}\ket\eta\bra\eta$. Measuring
an effect $R$ on one half of $\Phi_4$ prepares $R^T/4$ on the other half.
Comparison with the steered states $\rho_s/4$ gives
\[
W_B\rho_{\alpha_a}W_B^*=(P_a^0)^T\otimes\sigma,
\qquad
W_B\rho_{-\alpha_a}W_B^*=(Q_a^0)^T\otimes\sigma.
\]
Both reference projections are real symmetric, so these are the states
in \eqref{eq:RACnormal}. The ancillary state is the same for all eight
inputs because it is obtained from the single vector $\eta$ before
either row measurement.

Averaging within a parity class yields
$W_B\tau W_B^*=\Id_4/4\otimes\sigma$ and proves the support identity.
Since $\tau$ is Bob's reduced state, Theorem~\ref{thm:rigidity} shows
that the decoders leave $\supp\tau$ invariant and fixes them there. Extend $W_B$ from $\supp\tau$ to an
isometry $V$ of all $\mathcal H$ by enlarging $\mathcal K$ and mapping
the unused orthogonal complement into the new target dimensions. Extend
$\sigma$ by zero on the added ancillary space. All state equalities and
support-restricted decoder equalities are preserved. The ranks follow
from the tensor forms, and direct evaluation of the reference code proves
the converse.
\end{proof}

Thus, up to one isometry, every optimal finite-dimensional code consists of the same two
four-dimensional bases, one for each input parity, with an input-independent
ancillary state.
Its joint encoding support has dimension divisible by four. The ambient
encoding space can have additional unused dimensions.

The upper bound of Proposition~\ref{prop:RAC} also extends to encodings whose full input parity is
only approximately hidden. We apply the construction used
by Winter \cite[proof of Lemma~2, Eq.~(2), arXiv v6]{Winter} to the two
parity averages. Its ensemble version keeps the encoding space and decoders
fixed and gives an exact relation between the original and corrected scores.
A second, direct argument with the three decoding observables gives a
success bound strictly below one whenever the trace distance between the
two parity averages is less than one.

\begin{proposition}[Approximate full-parity hiding]\label{prop:approx-parity}
Let $\rho_s$, $s\in\{\pm1\}^3$, be arbitrary normalized states on a common
finite-dimensional Hilbert space, and let $F_{b|j}$, $b\in\{\pm1\}$,
be binary decoding POVMs. Define
\[
 \tau_\epsilon\coloneqq\frac14\sum_{\prod_j s_j=\epsilon}\rho_s,
 \qquad e\coloneqq\frac12\|\tau_+-\tau_-\|_1,
 \qquad p\coloneqq\frac1{24}\sum_{s,j}\Tr(\rho_sF_{s_j|j}),
\]
where $\epsilon\in\{\pm1\}$. There are states $\widehat\rho_s$ on the
same Hilbert space satisfying \eqref{eq:privacy} such that
\begin{equation}\label{eq:parity-correction}
 \frac12\|\widehat\rho_s-\rho_s\|_1\leq\frac{e}{1+e}
 \quad\text{for every }s,
 \qquad
 \widehat p=\frac{p+e/2}{1+e},
\end{equation}
where $\widehat p$ uses the original decoders $F_{b|j}$. Moreover,
\begin{equation}\label{eq:approx-parity-bound}
 p\leq\min\left\{\frac12+\frac{1+e}{\sqrt6},\,
 \frac{1+\cos\bigl(\frac13\arccos e\bigr)}2\right\}.
\end{equation}
In particular, if $p>\omega_Q(2,3)$, then
$e\geq\sqrt6\bigl(p-\omega_Q(2,3)\bigr)$.
\end{proposition}
\begin{proof}
Put $D\coloneqq\tau_+-\tau_-$ and
\[
 T_+\coloneqq\frac{|D|-D}{2},\qquad T_-\coloneqq\frac{|D|+D}{2},\qquad
 \widehat\rho_s\coloneqq\frac{\rho_s+T_{\prod_j s_j}}{1+e}.
\]
The operators $T_+$ and $T_-$ are the negative and positive parts of the
Hermitian operator $D$. Thus they are positive, and $\Tr D=0$ gives
$\Tr T_+=\Tr T_-=e$. The corrected operators are therefore normalized
states. Their parity averages agree because
\[
 \tau_++T_+=\tau_-+T_-
 =\frac{\tau_++\tau_-+|D|}{2}.
\]
For $\epsilon=\prod_j s_j$, positivity and the trace-norm triangle
inequality give
\[
 \frac12\|\widehat\rho_s-\rho_s\|_1
 =\frac{\|T_\epsilon-e\rho_s\|_1}{2(1+e)}
 \leq\frac{\Tr T_\epsilon+e\Tr\rho_s}{2(1+e)}
 =\frac{e}{1+e}.
\]
For each $\epsilon$ and $j$, exactly two strings of parity $\epsilon$
have $s_j=+1$ and two have $s_j=-1$. Since
$F_{+|j}+F_{-|j}=\Id$, the added contribution to the unnormalized
average success is
\[
 \frac1{24}\sum_{\epsilon=\pm1}\sum_{j=1}^3
 \Tr\bigl[T_\epsilon(2F_{+|j}+2F_{-|j})\bigr]=\frac e2.
\]
This proves the score identity in \eqref{eq:parity-correction}.
Proposition~\ref{prop:RAC} applies to the corrected code, so that
\[
 p\leq(1+e)\omega_Q(2,3)-\frac e2
 =\frac12+\frac{1+e}{\sqrt6}.
\]
Rearranging proves the asserted lower bound on $e$.

For the second bound, realize the three decoding measurements, as in
Section~\ref{sec:rectangle}, by one isometric embedding of the encoding
space and projective measurements with reflections $B_j$. This preserves
$p$ and $e$, and we suppress the embedding. Put
$p_j\coloneqq\frac18\sum_s\Tr\bigl(\rho_s(\Id+s_jB_j)/2\bigr)$, so
$p=(p_1+p_2+p_3)/3$. Choose normalized purifications $\psi_s$ on a common
space and form the unit vector and reflections
\[
 \Psi\coloneqq\frac1{\sqrt8}\bigoplus_s\psi_s,
 \qquad R_j\coloneqq\bigoplus_s s_j(B_j\otimes\Id).
\]
Then $\langle\Psi,R_j\Psi\rangle=2p_j-1$. Grouping the strings by parity
and using that $B_1B_2B_3$ is unitary gives
\[
 \operatorname{Re}\langle\Psi,R_1R_2R_3\Psi\rangle
 =\frac12\operatorname{Re}\Tr(DB_1B_2B_3)\leq\frac12\|D\|_1=e.
\]
The second bound increases with $e$ and equals $(2+\sqrt3)/4$ at $e=0$,
so it holds if $p\leq(2+\sqrt3)/4$. Otherwise $p_j\geq3p-2>1/2$ for each
$j$. The angle $\arccos\operatorname{Re}\langle u,v\rangle$ between unit
vectors is the geodesic distance on the unit sphere of the underlying
real Hilbert space; it obeys the triangle inequality and is invariant
under unitaries. Along $\Psi,R_1\Psi,R_1R_2\Psi,R_1R_2R_3\Psi$ the
consecutive angles are $\arccos(2p_j-1)$. Since $\arccos$ decreases and
is concave on $[0,1]$,
\[
 \arccos e\leq\arccos\operatorname{Re}\langle\Psi,R_1R_2R_3\Psi\rangle
 \leq\sum_j\arccos(2p_j-1)\leq3\arccos(2p-1).
\]
Dividing by three and applying the decreasing cosine on $[0,\pi]$ gives
$2p-1\leq\cos(\frac13\arccos e)$. Together with the first bound this
proves \eqref{eq:approx-parity-bound}.
\end{proof}

\section{Conclusion}\label{sec:conclusion}
The $2\times3$ magic rectangle has quantum value $(1+\sqrt{2/3})/2$.
Equality forces a maximally entangled pair of four-dimensional systems
and fixes the measurements on the local state supports in every
finite-dimensional optimal strategy, up to local isometries. The sum-of-squares certificate
proves the same value for commuting-operator strategies and the
almost-quantum relaxation; its equality relations determine the
measurement algebra underlying the finite-dimensional classification.

For finite-dimensional magic-square strategies, the rectangle determines
the exact score threshold for certifying randomness in Alice's complete
first row against quantum side information. At the threshold, every
finite-dimensional strategy with perfect prediction contains the same
four-dimensional entangled factor. The original measurements preserve
orthogonal subspaces indexed by Alice's first-row output. Separate local isometries extract
the factor from the whole purification while preserving coherence between
these output branches.

Above the threshold, the common pairwise Hamming distance of legal row
outputs gives a linear bound on the loss of predictability. Together with explicit
mixtures, it shows that one minus the best compatible guessing probability,
the smallest compatible conditional min-entropy and the smallest compatible
conditional von Neumann entropy all have linear order in the score excess.

The same equality relations classify all optimal finite-dimensional three-bit random access
codes that hide the full input parity. A common ancillary state and a
single isometry suffice for all eight original encodings, and the decoder
measurements are fixed on the encodings' occupied support. For approximate parity
hiding, an explicit correction on the same encoding space keeps all
decoders fixed and gives a dimension-independent bound linear in the
trace distance between the parity averages.

Extending the exact classifications to approximately optimal strategies and
encodings is a natural next problem. For the rectangle,
Corollary~\ref{cor:residual} supplies residual estimates controlled by the
score deficit; a possible route to quantitative local extraction is to
establish stability versions of the support and algebra arguments.
Further questions concern the exact guessing probability as a function of
the magic-square score, including the sharp coefficient near the
threshold, and the optimal coding success as a function of the trace distance
between the parity averages, including its slope at zero. For larger two-row rectangles, the
almost-quantum closed form conjectured by Adamson and Wallden remains a
natural target for analytic certificates.

\appendix
\section{The moment calculation with outcome projectors}\label{app:npa}
Corollary~\ref{cor:almost} uses the sum-of-squares identity inside a
truncated moment matrix. We verify here that its entire expansion uses
only the elementary constraints available at that level. The critical
point is that the legal row outcomes already imply the needed parity
reductions at degree two in outcome projectors.

Write Alice's formal outcome projectors as $A_{a|x}$, with
$a\in\{1,2,3,4\}$ and $x\in\{0,1\}$, and Bob's as $F_{b|j}$, with
$b\in\{\pm1\}$ and $j\in\{1,2,3\}$. The elementary relations are
\begin{align}\label{eq:outcome-relations}
A_{a|x}A_{c|x}&=\delta_{ac}A_{a|x},&
\sum_aA_{a|x}&=\Id,\notag\\
F_{b|j}F_{d|j}&=\delta_{bd}F_{b|j},&
\sum_bF_{b|j}&=\Id,
\end{align}
with self-adjoint outcome symbols and interchange of Alice and Bob
letters. The $1+AB$ word set is
\[
\mathcal V\coloneqq\{\Id\}\cup\{A_{a|x}\}\cup\{F_{b|j}\}
\cup\{A_{a|x}F_{b|j}\}.
\]
A feasible moment matrix has entries
$\Gamma_{u,v}\coloneqq L(u^*v)$ for $u,v\in\mathcal V$, is positive
semidefinite, satisfies $L(\Id)=1$, and obeys conjugate symmetry and
the zero, word-identification, and completeness constraints induced by
\eqref{eq:outcome-relations} \cite{NPA}. After ordering the parties,
the domain comprises words with at most two Alice letters and two Bob
letters. For example,
$A_1A_2F_1F_2=(A_1F_1)^*(A_2F_2)$, where the subscripted letters
here denote arbitrary individual outcome symbols. There are $63$
indices in the all-outcome description; eliminating one outcome per
measurement by completeness gives an equivalent $28$-index description.

Use the fixed labels \eqref{eq:alpha} to define
\[
D_j\coloneqq\sum_a\alpha_{aj}A_{a|0},\qquad
C_j\coloneqq-\sum_a\alpha_{aj}A_{a|1},\qquad
B_j\coloneqq\sum_{b=\pm1}bF_{b|j}.
\]
For distinct $i,k,j$, outcome orthogonality and completeness give the quadratic reductions
\begin{equation}\label{eq:quadratic-reductions}
D_j^2=C_j^2=B_j^2=\Id,
\qquad D_iD_k=D_j,\qquad C_iC_k=-C_j.
\end{equation}
For instance,
$D_iD_k=\sum_a\alpha_{ai}\alpha_{ak}A_{a|0}=D_j$
because each $\alpha_a$ has even parity. These reductions remain valid
inside moments with up to two letters of the other party, which covers
every context used in the certificate. Representative
moment constraints are
\begin{align*}
L(A_{a|x}A_{c|x}F_{b|i}F_{d|k})
&=\delta_{ac}L(A_{a|x}F_{b|i}F_{d|k}),\\
L(A_{a|x}A_{c|z}F_{b|j}F_{d|j})
&=\delta_{bd}L(A_{a|x}A_{c|z}F_{b|j}).
\end{align*}
In the first equality the left and right moments are respectively
$\Gamma_{A_{a|x}F_{b|i},A_{c|x}F_{d|k}}$ and
$\Gamma_{F_{b|i},A_{a|x}F_{d|k}}$, with the latter multiplied by
$\delta_{ac}$. In the second equality the right moment is
$\Gamma_{A_{a|x},A_{c|z}F_{b|j}}$. Completeness in these contexts is
likewise a sum of entries of the same matrix.

Each of $Z_j-q\Id$, $Y_j-V_j/r$, and $X_j+qB_j-qW_j$ in
\eqref{eq:SOS} belongs to $\operatorname{span}\mathcal V$.
The identities \eqref{eq:rowcancel} use at most two Alice letters;
\eqref{eq:cancela}--\eqref{eq:cancelc} use those identities with
zero, one, or two Bob letters. The diagonal square terms use
$B_j^2=\Id$ with at most two Alice letters as context. Thus every
step in the certificate expansion remains within the moment domain,
and applying $L$ gives
\begin{align*}
2r-L(S)={}&\frac{5r}{64}\sum_j L\bigl((Z_j-q\Id)^*(Z_j-q\Id)\bigr)\\
&+\frac r{12}\sum_j L\bigl((Y_j-V_j/r)^*(Y_j-V_j/r)\bigr)\\
&+\frac r{192}\sum_j
L\bigl((X_j+qB_j-qW_j)^*(X_j+qB_j-qW_j)\bigr)\geq0.
\end{align*}
The inequality follows from $\Gamma\succeq0$ term by term. It proves
the relaxed upper bound using only degree-two relations in each party.
Together with the reference quantum Gram matrix, this completes the
moment-level proof of Corollary~\ref{cor:almost}. The identification of
this bipartite relaxation with the almost-quantum set is given in
\cite[Lemma 3 and Remark 1]{almostquantum}.
\section{The tracial bound and its operator limitation}\label{app:trace}
An independent short proof of the rectangle bound is available when the
shared state is maximally entangled, but it cannot replace Theorem
\ref{thm:SOS}. It isolates the role of the trace and explains why a state-independent certificate is needed.
Write $\ket{\Phi_d}\coloneqq d^{-1/2}\sum_{a=1}^d\ket{a,a}$.

\begin{proposition}[Tracial bound]\label{prop:trace}
For $\Phi_d$ and arbitrary POVMs the rectangle score is at most
$(1+\sqrt{2/3})/2$. Equality is possible exactly when $4$ divides $d$,
and, up to a local unitary $\Omega\otimes\overline{\Omega}$, which fixes
$\Phi_d$, the attaining strategies are the reference strategy of Section~\ref{sec:rectangle} tensored with
$\Phi_{d/4}$.
\end{proposition}
\begin{proof}
Let $P_a,Q_a$ be the two physical row POVMs, both labeled by $\alpha_a$.
After negating the second row, the coordinate observables
$D_j\coloneqq\sum_a\alpha_{aj}P_a$ and $C_j\coloneqq-\sum_a\alpha_{aj}Q_a$ are contractions.
The identity $\alpha_a\cdot\alpha_b=4\delta_{ab}-1$ gives
\[
\sum_j\Tr(D_jC_j)=d-4\sum_a\Tr(P_aQ_a)\leq d.
\]
Thus $\sum_j\Tr(X_j^2)\leq8d$, while
$\sum_j\Tr(B_j^2)\leq3d$. Hilbert--Schmidt Cauchy--Schwarz bounds the
Bell sum by $\sqrt{8d\,3d}/d=2\sqrt6$.
Equality forces all $D_j$, $C_j$ and $B_j$ to be reflections and
$B_j^T=X_j/q$. Each coordinate marginal is then projective, so
$P_a$ is dominated by the projection $(\Id+\alpha_{aj}D_j)/2$ for every $j$; distinct $a,b$ differ in
some coordinate, hence $P_aP_b=0$, and the $P_a$ are the joint spectral
projections of the $D_j$; likewise for the $Q_a$. Since $B_j^2=\Id$ gives
$X_j^2=q^2\Id$, we get $\ac{D_j}{C_j}=2\Id/3$, and
Lemma \ref{lem:conference} supplies the classification and rank condition.
\end{proof}

The operator inequality
$\sum_j\ac{D_j}{C_j}\leq2\Id$ is false. Its trace is the bound on $\sum_j\Tr(D_jC_j)$ above, and, through $\sum_jX_j^2\leq8\Id$, it would give the same Cauchy--Schwarz argument in every state. Keep the diagonal $D_j^0$ of Section~\ref{sec:rectangle}, let $J_4$ be the all-ones matrix and $H\coloneqq\Id-J_4/2$, and set
$D_j\coloneqq D_j^0$ and $C_j\coloneqq-HD_jH$. Both parity constraints hold. The identity $\alpha_a\cdot\alpha_b=4\delta_{ab}-1$ gives $\sum_jD_jHD_j=\Id+J_4/2$, so
\[
\sum_j\ac{D_j}{C_j}=2\Id-4H=-2\Id+2J_4
\]
has eigenvalues $6,-2,-2,-2$. Each $\Tr(P_aQ_a)$ is nonnegative for the
spectral projections $P_a,Q_a$ of these $D_j$ and $C_j$. The operator
$\sum_a\ac{P_a}{Q_a}=\frac14\bigl(2\Id-\sum_j\ac{D_j}{C_j}\bigr)$
equals $H$ and has eigenvalue $-1$.

\section{A general projection bound}\label{app:projection}
A projection onto Eve's predicted key sector gives a quantitative bound
whenever the corresponding threshold for perfect prediction is known.
For the magic-square game this yields the following weaker consequence,
independently of the replacement argument of Theorem~\ref{thm:linear}.
\begin{proposition}[Projection estimate]\label{prop:gentle}
If a strategy as in Theorem~\ref{thm:linear} has score $w\geq\beta_0$, then
\[
\pg(K\mid E,x=0)\leq1-(w-\beta_0)^2.
\]
Consequently $H_{\min}(K\mid E,x=0)\geq
-\log_2[1-(w-\beta_0)^2]$.
\end{proposition}
\begin{proof}
Use Eve's optimal guessing measurement and keep its outcome as a classical
register; denote the resulting joint state by $\rho$.
Dilate all of Alice's measurements with one local embedding, as in
Section~\ref{sec:rectangle}; her key measurement then has projections $P_k$.
Let $P$ be the block-diagonal projection onto Alice's predicted key outcome
in each Eve sector. If $p=\pg(K\mid E,x=0)$, then $\Tr(P\rho)=p$.
Since $K$ has four possible values, $p\geq1/4$, and the normalized state $\rho'\coloneqq P\rho P/p$ has a perfectly guessed
key. Together with the same measurements and with Eve's register as side information, it forms one of the strategies that define $\beta_0$, so its score is at most $\beta_0$.
A purification projected and normalized in the same way has overlap
$\sqrt p$ with the original purification. Hence
$\tfrac12\|\rho-\rho'\|_1\leq\sqrt{1-p}$.
The game winning effect lies between zero and identity, and therefore
$w\leq\beta_0+\sqrt{1-p}$. Rearranging proves the claim.
\end{proof}

\section{A second sum-of-squares identity}\label{app:second}
The rectangle gap also has a decomposition using the total Bell residual and pairwise differences of its three summands. With the notation of Section~\ref{sec:sos},
\begin{align*}
2r\Id-S={}&\frac r{32}(S-2r\Id)^2
+\frac r{48}\sum_{i<j}(Z_i-Z_j)^2
+\frac r{96}\Bigl(\sum_jY_jB_j\Bigr)^2\\
&+\frac r{16}\sum_j(Y_j-V_j/r)^2
+\frac r{96}\sum_j(W_j-2B_j)^2.
\end{align*}
To verify the identity, put $T\coloneqq\sum_jY_jB_j$. In addition to
\eqref{eq:rowcancel}--\eqref{eq:cancelc}, direct expansion gives
\[
 S^2+T^2=12\Id+2J,\qquad
 \sum_{i<j}(Z_i-Z_j)^2=3\sum_jX_j^2-S^2.
\]
Indeed, the first identity uses $X_j^2+Y_j^2=4\Id$ and, for each ordered
pair, $X_iX_k+Y_iY_k=2Y_j$; the second is the sum of three squared differences.
Multiplying the right side of the decomposition by $96/r$ gives
\[
3(S-2r\Id)^2+
2\sum_{i<j}(Z_i-Z_j)^2+T^2+
6\sum_j(Y_j-V_j/r)^2+\sum_j(W_j-2B_j)^2.
\]
The first three terms reduce to
$84\Id+6\sum_jX_j^2+2J-12rS$. The last two reduce to
$108\Id-6\sum_jX_j^2-2J-4rS$ by
\eqref{eq:rowcancel}--\eqref{eq:cancelc}, using $r^2=6$.
Their sum is $192\Id-16rS=(96/r)(2r\Id-S)$, as required.
Both certificates are identities in the algebra generated by the two
even-parity triples $D_j,-C_j$ and the three Bob reflections; within
either party the order of different measurement settings is retained.
\pdfbookmark[1]{References}{references}

\end{document}